\documentclass[a4paper,USenglish,modulo,thm-restate]{lipics-v2019}

\usepackage{xspace}

\usepackage[ruled,vlined,noend]{algorithm2e}

\usepackage{placeins}

\newcommand{\Frechet}[0]{Fr\'echet\xspace}
\newcommand{\Seq}[1]{\langle #1 \rangle}
\newcommand{\Set}[1]{\lbrace #1 \rbrace}
\newcommand{\TrajSet}{\mathcal{T}}
\newcommand{\RoadNet}{\mathcal{G}}

\newcommand{\Reals}{\mathbb{R}}
\newcommand{\Solution}{\mathcal{P}}

\newcommand{\etal}{\textit{et al.}\xspace}

\title{Route Reconstruction from Traffic Flow\\  via Representative Trajectories} 

\titlerunning{Route Reconstruction from Traffic Flow via Representative Trajectories}

\author{Bram Custers}{TU Eindhoven, the Netherlands }{b.a.custers@tue.nl}{https://orcid.org/0000-0001-9342-319X}{Supported by HERE Technologies and the Dutch Research Council (NWO); 628.011.005}
\author{Wouter Meulemans}
{TU Eindhoven, the Netherlands }{w.meulemans@tue.nl}{https://orcid.org/0000-0002-4978-3400}{}
\author{Bettina Speckmann}{TU Eindhoven, the Netherlands }{b.speckmann@tue.nl}{https://orcid.org/0000-0002-8514-7858}{Supported by the Dutch Research Council (NWO); 639.023.208.}
\author{Kevin Verbeek}{TU Eindhoven, the Netherlands }{k.a.b.verbeek@tue.nl}{https://orcid.org/0000-0003-3052-4844}{}

\authorrunning{B. Custers, W. Meulemans, B. Speckmann, and K. Verbeek} 

\Copyright{Bram Custers, Wouter Meulemans, Bettina Speckmann, and Kevin Verbeek} 



\category{} 

\relatedversion{} 

\supplement{Implementation of our algorithms is available at \href{https://github.com/tue-alga/RouteReconstruction}{https://github.com\discretionary{/}{}{/}tue-alga\discretionary{/}{}{/}RouteReconstruction}}

\acknowledgements{The authors would like to thank HERE Technologies for providing the HR dataset.}

\nolinenumbers 

\hideLIPIcs  

\EventEditors{John Q. Open and Joan R. Access}
\EventNoEds{2}
\EventLongTitle{Symposium on Computational Geometry (SOCG 2021)}
\EventShortTitle{SOCG 2021}
\EventAcronym{SOCG}
\EventYear{2021}
\EventDate{December 24--27, 2021}
\EventLocation{Buffalo, New York, United States of America}
\EventLogo{}
\SeriesVolume{42}
\ArticleNo{23}

\begin{document}

\maketitle

\begin{abstract}
Understanding human mobility patterns is an important aspect of traffic analysis and urban planning. 
Trajectory data provide detailed views on specific routes, but typically do not capture all traffic. 
On the other hand, loop detectors built into the road network capture all traffic flow at specific locations, but provide no information on the individual routes. Given a set of loop-detector measurements as well as a (small) set of representative trajectories, our goal is to investigate how one can effectively combine these two partial data sources to create a more complete picture of the underlying mobility patterns. Specifically, we want to reconstruct a realistic set of routes from the loop-detector data, using the given trajectories as representatives of typical behavior. 

We model the loop-detector data as a network flow that needs to be covered by the reconstructed routes and we capture the realism of the routes via the strong \Frechet distance to the representative trajectories. We prove that several forms of the resulting algorithmic problem are NP-hard. Hence we explore heuristic approaches which decompose the flow well while following the representative trajectories to varying degrees. First of all, we propose an iterative \emph{\Frechet Routes} (FR) heuristic which generates candidates routes which have bounded \Frechet distance to the representative trajectories. 
Second we describe a variant of multi-commodity min-cost flow (MCMCF) which is only loosely coupled to the trajectories. Lastly we also consider global min-cost flow (GMCF) which is essentially agnostic to the representative trajectories.

We perform an extensive experimental evaluation of our two proposed approaches in comparison to the min-cost flow baseline, both on synthetic and on real-world trajectory data. To create a ground truth for our experiments, we extract the flow information from map-matched trajectories. We find that GMCF explains the flow best, but produces a large number of routes (significantly more than the ground truth); these routes are often nonsensical. Also MCMCF produces a large number of routes which explain the flow reasonably well, however, the routes are mostly realistic. In contrast, FR produces significantly (orders of magnitude) smaller sets of realistic routes which still explain the flow well, albeit at the cost of a higher running time.
\end{abstract}

\newpage
\section{Introduction}
\label{sec:introduction}

Understanding human mobility patterns is an important aspect of traffic analysis and urban planning. To analyze mobility, we need to answer questions about various aspects of (vehicle) traffic, such as ``How busy is this road at different times of the day?'', ``Between which locations do people travel most frequently?'' or ``Which routes do people use to travel through the road network''? To answer such questions we can use various heterogeneous data sources, which generally fall into one of two categories: checkpoint data and tracking data.

Checkpoint data originate from measurements by static devices such as loop detectors or traffic cameras, placed on fixed locations throughout the road network. They provide a comprehensive view of the amount of traffic flow at that particular location, but inherently no information on how people navigate through the network. Tracking data, on the other hand -- predominantly captured through GPS in smart phones and navigation systems -- provide a detailed view of individual behavior in the form of trajectories. However, trajectory data does not describe the general traffic flow, as not all vehicles are tracked or tracked by the same system. Furthermore, the (often significant) detail in trajectories also raises privacy concerns, so trajectory data are frequently segmented and anonymized before analysis.

Both categories thus give only a partial view on the full dynamics of human mobility. Given their complementarity it is natural to investigate possibilities for data fusion, that is, data enrichment that leverages the strength of both. Given loop-detector measurements as well as a (small) set of representative trajectories, we investigate how one can effectively combine these two partial data sources to create a more complete picture of the underlying mobility patterns. Specifically, we want to reconstruct a realistic set of routes from the loop-detector data, using the given trajectories as representatives of typical behavior.

\subparagraph{Contributions and organization}
After a brief review of related work, we formally model our problem in Section~\ref{sec:modelling}, while also introducing the necessary concepts and notation. We arrive at a formal problem statement which models the loop-detector data as a time-independent network flow that needs to be covered by the reconstructed routes; we capture the realism of the routes via the strong \Frechet distance to the representative trajectories.
In Section~\ref{sec:hard} we prove that several forms of the resulting algorithmic problem are NP-hard even in restricted settings.
Hence we explore heuristic approaches which decompose the flow well while following the representative trajectories to varying degrees. In Section~\ref{sec:algo}, we propose an iterative \emph{\Frechet Routes} (FR) heuristic which generates candidates routes which have bounded \Frechet distance to the representative trajectories. 
In the same section we also describe a variant of multi-commodity min-cost flow (MCMCF) which is only loosely coupled to the trajectories. 

In Section~\ref{sec:experiments} we report on an  experimental evaluation of these proposed approaches in comparison to a global min-cost flow baseline
(GMCF) which is essentially agnostic to the representative trajectories. To make meaningful claims in terms of quality, we derive a ground truth by map matching real-world trajectories, and by generating synthetic routes in a real-world road network.
We find that GMCF explains the flow best, but produces a large number of often nonsensical routes (significantly more than the ground truth). 
MCMCF produces a large number of mostly realistic routes which explain the flow reasonably well. In contrast, FR produces significantly 
smaller sets of realistic routes which still explain the flow well, albeit at the cost of a higher running time. We note that our approaches 
are not restricted to loop-detector data.
We discuss avenues for further research in Section~\ref{sec:discussion}.

\subparagraph{Related work} 
Our problem is closely related to flow decomposition, where the goal is to decompose an aggregated flow into paths, optimizing a given objective function. Any flow can be decomposed into at most $O(E)$ paths and cycles, where $E$ is the number of edges in the graph~\cite{ahuja2017network}. However, given a set of paths that decompose the flow, it is NP-hard to determine the correct integral coefficients for each path~\cite{kloster2018practical}. Minimizing the number of paths in a decomposition is also NP-hard~\cite{vatinlen2008simple}, and thus various approximation algorithms have been developed~\cite{hartman2012split}.
We require that the reconstructed routes are similar to one of the representative trajectories, but we do not require that the flow is explained completely.  

We are aware of only a single geometric approach to reconstruct flow from checkpoint data. Duckham~\etal~\cite{duckham2016modeling} use the Earth Mover's distance to estimate the movement of couriers from checkpoint data. Given the limited data, the results are quite accurate, but a full reconstruction of the movement is clearly out of reach. 

Reconstructing a route (in a network) given a GPS trajectory is referred to as map matching \cite{alt2003matching,HaunertB-2012,jagadeesh2017online,WenkSP-2006}, see also the survey by Quddus~\etal~\cite{quddus2007current}. The goal is to find a route in a given network, accounting for potential misalignment between GPS measurements and the network, noise inherent in GPS systems, and inaccuracies in the road network.
There are a wide variety of available algorithms; particularly prominent are solutions based on hidden Markov models, a strategy which was pioneered by Newson and Krumm~\cite{newson2009hidden}.
Of specific relevance to our work is the result by Alt~\etal~\cite{alt2003matching} which solves map matching under the \Frechet distance (see Section~\ref{sec:algo}). Specifically, their algorithm decides in quadratic time whether a graph admits a path with
\Frechet distance at most $\varepsilon$ to an input trajectory.
Generally, map-matching techniques are not designed to explain flow data, but rather focus on correcting measurement errors in a single trajectory.
Moreover, if we insist on simple routes, map matching is NP-hard for various measures, including the \Frechet distance~\cite{loffler2017discretized} and Hausdorff distance~\cite{bouts2016mapping}. If the road network is a perfect grid, routes with bounded (but not necessarily minimal) \Frechet and Hausdorff distance can be found efficiently~\cite{bouts2016mapping}.

\section{Modeling}\label{sec:modelling}

Our input has three components: $(1)$ a road network, given as a graph $\RoadNet = (V,E)$; $(2)$ a set of representative trajectories $\TrajSet$, each encoded by a sequence $\langle p_1,\ldots,p_n \rangle$ of measurements; and $(3)$ loop-detector data, which are traffic-volume measurements expressed as the number of cars at a specific time at a specific location. Our goal is to combine these heterogeneous data sources to create a more complete picture of the underlying mobility patterns. Specifically, we want to reconstruct a realistic set of routes from the loop-detector data, using the trajectories as representatives of typical behavior. Here we discuss the modeling decisions we made to finally arrive at a formal problem statement which is amenable to algorithmic treatment. While doing so we also introduce our notation.

\subparagraph{Road network}
A road network is specified as a graph, $\RoadNet = (V,E)$ with $n$ vertices and $m$ edges. We assume that each vertex $v \in V$ has a position in $\Reals^2$ and that the edges are straight line segments between their endpoints. The graph is directed and thus bi-directional roads are represented by two separate edges which lie on top of each other. 

\subparagraph{Trajectories} 
A trajectory is a curve in space-time that represents a route taken by a vehicle. 
A trajectory $T$ is represented as a sequence $\langle p_1,\ldots,p_n \rangle$ of measurements $p_i$, where $p_i$ is a tuple $(x_i,y_i,t_i)$ containing the coordinates $(x_i,y_i) \in \Reals^2$ and the timestamp $t_i$ of the measurement. 
Trajectories obtained from vehicles are often \emph{map-matched} to a road network $\RoadNet = (V,E)$: the spatial components of a trajectory are matched to points on the road network, and a reasonable route in $\RoadNet$ is reconstructed for the trajectory (the definition of ``reasonable'' depends on the modeling decision made by the respective algorithm). The map-matched trajectory can then be expressed by a \emph{route} $P = \Seq{e_1,\ldots,e_k}$: a sequence of edges that encode the traversed path in the road network $\RoadNet$. We say that a route is \emph{simple} if it visits every vertex in $\RoadNet$ at most once.  
For ease of notation, we introduce the function $M(P,e)$ that indicates how often the edge $e \in E$ is traversed in the route $P$. Note that $M(P, e)$ is $0$ or $1$ if $P$ is simple (necessary but not sufficient), but may take on higher values if $P$ is not simple. 

\subsection{Modeling loop-detector data as time-independent complete flows}
Loop-detector data is gathered by counting the number of vehicles passing an induction loop in the road network. These counts are then aggregated over a fixed time interval. These time intervals generally range from minutes to hours and thus can give an accurate view on the traffic volume within a single day at that particular location in the road network.

Mathematically speaking, loop-detector data can be interpreted as an (incomplete) {\bf flow} on the road network: we assign the aggregate loop-detector data to the edge in the road network where the detector is located and interpret the data as the volume of traffic into and out of this edge, that is, the flow through this edge. As not every edge in a road network has a loop detector, the flow data is a priori incomplete. Reconstructing a complete flow is challenging due to the inherent uncertainty surrounding the exact driven routes between detector locations, see Castillo~\etal~\cite{castillo2008observability}. We focus on reconstructing routes from {\bf complete flow} information; reconstructing such a complete flow is beyond the scope of this paper.

Loop-detector data are inherently dependent on time and hence a priori imply a time-dependent flow. However, time-dependent flows pose several data, modeling, and complexity problems. 
First of all, we need additional data to model the time needed to traverse the network. We need to know the travel times for edges to be able to reason about realistic routes in the network for the flow. The realism of the reconstructed routes then heavily depends on the accuracy of these values and on time itself. Furthermore, time-dependent flows naturally require the use of time-dependent representative trajectories, in which case the set of representative trajectories will generally be (too) sparse. In addition, as we show in Section~\ref{sec:hard}, time-independent formulations of our problem are already computationally hard, which suggests that the time-dependent problem will be even harder to compute. Hence we model the loop-detector data via \textbf{time-independent} flow, which one can interpret as a ``long-time average'' of the loop-detector measurements.

We formalize the flow as follows: given a network $\RoadNet = (V,E)$, a \emph{flow} $(f,S,T)$ on $\RoadNet$ is specified by a flow function $f \colon E \rightarrow \Reals_{\ge 0}$ mapping each edge of $\RoadNet$ to its flow value. In addition, there are sources $S \subseteq V$ and sinks $T \subseteq V$ for the flow. Commonly, a flow function must satisfy the flow-conservation property: for each vertex that is neither a source nor a sink, the sum over all incoming flow is equal to the sum over all outgoing flow.
However, errors in actual data may cause violations of the flow-conservation property. Furthermore, traffic volume generally does not specify sources or sinks. We hence introduce the notion of a \emph{flow field} to describe the measured traffic volume, which is a function $\phi\colon E\rightarrow \mathbb{N}$ (not necessarily satisfying the flow-conservation property). 

\subsection{Reconstructed routes}
\label{ssec:reconstructedroutes}

Our goal is to compute a multiset $\overline{\Solution}$ of realistic routes which explains a flow field $\phi$ well. Recall that a route $P$ is a sequence of connected edges $\Seq{e_1,\ldots,e_l}$ in $\RoadNet$ and the multiplicity $M(P,e)$ is the number of occurrences of $e$ in $P$. We represent the multiset $\overline{\Solution}$ by a base set of routes $\Solution$ (a \emph{basis} $\Solution$ for short), along with associated frequency counts. It seems natural to assume that these counts should have integer values. However, with that restriction, even computing the correct counts for a specific base set of routes to explain a given flow is NP-hard~\cite{kloster2018practical}.  We therefore relax the counts $c \colon \Solution \rightarrow \Reals_{\ge 0}$ to be fractional \emph{coefficients}. This relaxation allows us to efficiently compute the coefficients for a specific base set and flow.

We say that the real-valued multiset of routes $\overline{\Solution}=(\Solution,c)$ is a \emph{reconstruction} of the flow field $\phi$. A reconstruction $(\Solution,c)$ defines
a flow $(f_{(\Solution,c)},S_{\Solution},T_{\Solution})$ as follows:
\begin{equation*}
    \forall e\in E: \qquad f_{(\Solution,c)}(e) = \sum_{P \in \Solution} M(P,e) c(P)
\end{equation*}
where $S_{\Solution}$ and $T_{\Solution}$ are the sets of start and end vertices of the routes in $\Solution$, respectively. We need to quantify how well the reconstruction $(\Solution,c)$, represented by the flow $(f_{(\Solution,c)},S_{\Solution},T_{\Solution})$, explains the input flow field $\phi$. The error in the loop-detector measurements can be positive or negative, suggesting a measure based on the absolute difference between $(f_{(\Solution,c)},S_{\Solution},T_{\Solution})$ and $\phi$ per edge. In line with traffic-analysis literature~\cite{cao2013bilevel,cascetta1984estimation}, we compute the \emph{flow deviation} $\Delta(\Solution,c,\phi)$ as the sum of squared differences between $(f_{(\Solution,c)},S_{\Solution},T_{\Solution})$ and $\phi$ over all edges.
\begin{equation*}
 \Delta(\Solution,c,\phi) = \sum_{e \in E} (\phi(e) - f_{(\Solution,c)}(e))^2
\end{equation*}
We note that for a fixed base set $\Solution$ of routes, the coefficients $c$ that minimize the flow deviation can be computed efficiently with standard techniques~\cite{bro1997fast,kim2010tackling}.

\subsection{Realistic routes}
\label{ssec:realism}

We are given a set of trajectories $\TrajSet$ that represent typical behavior of vehicles in the road network. Our aim is to compute a realistic reconstruction of the flow field based on $\TrajSet$. We measure the realism of a route in the basis $\Solution$ via its distance to the closest trajectory in $\TrajSet$.

There are many possible similarity measures for polylines, amongst which the \Frechet distance~\cite{alt1995computing}, the Hausdorff distance, and Dynamic Time Warping~\cite{berndt1994using}. In our setting we might encounter a large difference (either way) in spatial resolution between the road network and the representative trajectories, since the sampling densities of vehicle trajectories differ greatly between providers and sampling technologies used. In the presence of such large discrepancy in sampling density, discrete measures such as Dynamic Time Warping and the discrete versions of the \Frechet distance and Hausdorff distance are known to perform poorly, since measurements have to be matched to vertices of the road network.

The Hausdorff distance and the weak \Frechet distance do not capture the order of points and edges in trajectories and are hence less suitable for our purpose. Hence we choose the strong \Frechet distance to measure the realism of our reconstructed routes: it naturally captures the variability in the paths while encouraging that the general direction of reconstructed routes and representative trajectories are similar.

The (strong) \emph{\Frechet distance} $d_F(P,Q)$ between two curves $P, Q\colon[0, 1] \rightarrow \Reals^2$ is defined as
\begin{equation*}
d_F(P,Q) = \inf_{\alpha,\beta} \sup_{t\in[0,1]} ||P(\alpha(t)) - Q(\beta(t))||,
\end{equation*}
where $\alpha$ and $\beta$ are reparameterizations of $P$ and $Q$, respectively.
The functions $\alpha,\beta$ must be strictly monotonically increasing, with $\alpha(0) = \beta(0) = 0$ and $\alpha(1) = \beta(1) = 1$. 
To determine if two curves lie at \Frechet distance at most $\varepsilon$, we use the so-called 
\emph{free-space diagram}. This diagram represents matching locations on curves $P$ and $Q$ that are within $\varepsilon$ Euclidean distance of each other. Two curves have \Frechet distance at most $\varepsilon$ if a strictly monotone path can be found through this diagram. This problem can be solved in $O(n^2)$ time where $n$ is the total complexity of the two curves \cite{alt1995computing}. 

We say that a route is \textbf{realistic} if it lies within a prespecified \Frechet distance $\varepsilon$ of the closest representative trajectory. The parameter $\varepsilon$ controls the realism of our reconstruction. We say that a reconstruction $(\Solution,c)$ is realistic if all routes $P \in \Solution$ are realistic.

The routes taken by vehicles tend to be simple, since humans generally take shortest paths to their destinations. Hence we would prefer to reconstruct simple routes only. However, as we show in Section~\ref{sec:hard}, even just minimizing the flow deviation is NP-hard for simple paths. Our heuristic approaches hence prefer simple routes but do not exclude non-simple ones.

\subsection{Formal problem statement}\label{subsec:formal}
Our complete input is the road network $\RoadNet = (V,E)$, the set of representative trajectories $\TrajSet$, the flow field $\phi$ induced by the loop-detector data, and realism parameter $\varepsilon > 0$.
We want to find a realistic reconstruction $(\Solution,c)$ such that the flow deviation $\Delta(\Solution,c,\phi)$ is minimized.

\section{Computational complexity}
\label{sec:hard}

In this section we explore the computational complexity of our problem. First of all, we restrict the reconstructed routes to be simple. In this setting, even computing just a single realistic route is NP-hard~\cite{loffler2017discretized}. By extension, computing a realistic simple reconstruction is hard as well. Hence, we next consider a restricted variant of the problem, where the reconstructed routes do not have to be realistic and even share start and end point. Specifically, we require that all routes in the reconstruction are simple and start at a vertex $s$ and end at a vertex $t$. We refer to such a reconstruction as an $(s,t)$-reconstruction. However, we show that even this simplified problem is NP-hard. For this result, we consider two variants of the deviation function: the sum of squared differences as defined in the previous section, but also the sum of absolute differences. In the following two theorems we refer to these as \emph{squared} and \emph{absolute} deviation, respectively (proofs can be found in Appendix~\ref{app:hard}).
\begin{restatable}{theorem}{absdeviation}
\label{thm:abshard}
Let $\RoadNet$ be a road network with source $s$ and sink $t$, and let $\phi$ be an associated flow field. It is NP-hard to compute an $(s,t)$-reconstruction with only simple paths that minimizes the absolute deviation to $\phi$.
\end{restatable}
\begin{restatable}{theorem}{sqrhard}
\label{thm:sqrhard}
Let $\RoadNet$ be a road network with source $s$ and sink $t$, and let $\phi$ be an associated flow field. It is NP-hard to compute an $(s,t)$-reconstruction with only simple paths that minimizes the squared deviation to $\phi$. 
\end{restatable}

Consequently, we weaken the requirements even further and study relaxed $(s,t)$-recon\-struc\-tions which may contain non-simple routes. 
We sketch a simple algorithm which approaches a relaxed $(s,t)$-reconstruction with optimal flow deviation. First, we find the min-cost flow under the flow deviation~\cite{vegh2016strongly}. Then we construct a so-called pathflows-cycleflows decomposition of the resulting flow by employing a regular flow decomposition algorithm~\cite{ahuja2017network}. Now we can construct non-simple routes from the cycles by routing a shortest $(s,t)$-path via each cycle.
By sending a very small amount of flow along these routes and spinning around the cycles many times, we can construct non-simple route flows that come arbitrarily close to the cycle flows they were constructed from. Hence, they also come arbitrarily close to the optimal deviation. The resulting routes are highly non-simple and generally nonsensical.

We now return to the problem as stated in Section~\ref{subsec:formal}: we reintroduce the requirement that the reconstructed routes must be realistic, while still allowing non-simple routes and dropping the $(s,t)$-requirement. In this setting, minimizing the flow deviation becomes trivial if $\varepsilon$ is large enough: we can use single edges as routes in our reconstruction, thus trivially covering the entire flow field. 

For smaller values of $\varepsilon$, as of yet we cannot establish the computational complexity. We observe that, already for simple paths, the problem under the absolute and squared deviations is always convex; the difficulty stems from the exponential number of candidate paths to consider for the reconstruction. Furthermore, by Carath\'{e}odory's Theorem, we know that the optimal reconstruction contains at most $|E|$ paths. Thus, the difficulty of the problem lies in efficiently searching through the solution space for the best routes. Allowing non-simple routes grows the solution space, but may potentially make it possible to search the space more efficiently. However, we see no reason why this would be the case if we additionally require the routes to be realistic. We therefore conjecture that the problem is also NP-hard for non-simple realistic reconstructions.

\section{Route reconstruction algorithms}
\label{sec:algo}

In this section we present heuristic approaches which decompose the flow well while following the representative trajectories to varying degrees. In Section~\ref{subsec:FR} we describe our iterative \Frechet Routes (FR) heuristics which generates only realistic routes, that is, routes which have bounded \Frechet distance to the representative trajectories. We discuss two variants of this heuristic (WFR and EFR) which differ in their approach to selecting realistic routes. In Section~\ref{subsec:mincost} we scribe a  variant of multi-commodity min-cost flow (MCMCF) which is loosely coupled to the representative trajectories, as well as, the global min-cost flow (GMCF) baseline which is essentially agnostic to the representative trajectories.

\subsection{\Frechet Routes}\label{subsec:FR}

Recall that our goal is to find a realistic basis $\Solution$ and coefficients $c$ such that the flow deviation is minimized. Our \Frechet Routes (FR) heuristic decouples finding the basis and deciding on the coefficients $c$ for a given basis. We grow the basis iteratively, aiming to improve the deviation after each iteration (see Algorithm~\ref{alg:columngeneration}). Since the solution space is infinite, the main challenge is to find a basis that is small enough for efficient computation but comprehensive enough to result in a small deviation from the flow field. 

\begin{algorithm}[t]
\SetAlgoLined
\KwData{Road network $\RoadNet$, flow field $\phi$,  representative trajectories $\TrajSet$, threshold $\varepsilon$, and maximum number of iterations $i_\text{max}$}
\KwResult{Decomposition $(\Solution,c)$}
 Initialize residual flow field $\phi_r \leftarrow\ \phi$;
 basis $\Solution\leftarrow\ \emptyset$;
 coefficients $c\leftarrow\ \emptyset$\;
 \Repeat{$i_\text{max}$ iterations have been performed}{
    \For{$T \in \TrajSet$}{
        $\Solution \leftarrow\ \Solution \cup$ \textsc{GenerateBasisRoutes}($T$, $\phi_r$,$\RoadNet$,$\epsilon$)\;
    }
    $(\phi_r,c) \leftarrow\ $\textsc{ComputeCoefficients}($\Solution$,$\phi$)\;
    $(\Solution,c) \leftarrow\ $\textsc{Prune}($\Solution$,$c$)\;
 }
 \Return{$(\Solution,c)$}
 \caption{\textsc{FrechetRoutes}($\RoadNet$, $\phi$, $\TrajSet$, $\varepsilon$, $i_\text{max}$)}
 \label{alg:columngeneration}
\end{algorithm}

In one iteration of our heuristic we add new routes to the basis for each representative trajectory independently (\textsc{GenerateBasisRoutes}) and evaluate the resulting basis.
For a given basis we can compute the coefficients $c$ that minimize the flow deviation efficiently with standard techniques~\cite{bro1997fast,kim2010tackling}. We prune the basis by eliminating duplicate routes and routes with coefficient zero.
Furthermore, we compute the \emph{residual} flow field $\phi_r\colon E\rightarrow \Reals$, defined as $\phi_r(e) = \phi(e) - \sum_{P \in \Solution} M(P,e) c(P)$ for all edges $e\in E$. The residual flow field guides our search for new basis elements.

\subparagraph{Generating basis routes} Given the road network $\RoadNet$, a single representative trajectory $T$, the residual flow field $\phi_r$ (or the flow field $\phi$ in the first iteration), and threshold $\varepsilon$ on the \Frechet distance, we generate basis routes for $T$ as follows. The residual flow field stems from our current reconstruction $(\Solution,c)$. This reconstruction has an associated deviation $\Delta(\Solution,c,\phi)$ (this deviation is the flow within $\phi$ initially). If we extend the basis $\Solution$ with a path $P$ and an associated positive coefficient $c_P$, the deviation changes by: 
\begin{eqnarray*}
\Delta(\Solution \oplus P, c \oplus c_P,\phi) - \Delta(\Solution, c, \phi)
=&\sum_{e\in P} (\phi_r(e)-M(P,e)c_P)^2 - \phi_r(e)^2\\
=&-c_P\sum_{e\in P} M(P,e)(2 \phi_r(e) - M(P,e)c_P)
\end{eqnarray*}
For routes that visit an edge at most once (e.g.,  simple routes), the above simplifies to $-c_P \sum_{e \in P} (2 \phi_r(e) - c_P)$.
Negative values reduce deviation, so we are particularly interested in capturing edges in routes which have high residual flow $\sum_{e \in P} \phi_r(e)$.
Below we describe two different approaches to do so:  Edge-inclusion \Frechet Routes (EFR), which selects $k$ new routes per trajectory, and Weighted \Frechet Routes (WFR), which select one new route per trajectory. Both are adaptations of \Frechet map-matching as described by Alt~\etal~\cite{alt2003matching}.

\subparagraph{\Frechet map-matching}
Here we briefly sketch the algorithm by Alt~\etal, additional details can be found in Appendix~\ref{app:map-matching}. The input is a trajectory $T=\Seq{p_1,\ldots,p_{l}}$, a road network $\RoadNet$, and a threshold $\varepsilon$. The output is a route in $\RoadNet$ with \Frechet distance at most~$\varepsilon$ to $T$. 
The algorithm uses a \emph{free-space manifold}: a generalization of the free-space diagram that represents matching locations within distance $\varepsilon$ between $T$ and locations in $\RoadNet$. A monotone path in this manifold corresponds to a route in $\RoadNet$ that is within $\varepsilon$ \Frechet distance of $T$. We can determine the existence of such a path by applying a mixture of Dijkstra and a sweepline. Intervals of free space at vertices are \emph{white intervals}, and the algorithm determines their subintervals that are \emph{reachable} with a monotone path. The sweepline processes such reachable intervals in order of their lower endpoint, and follows edges $e$ to the target vertex via monotone paths to find new reachable intervals.

\subparagraph{Edge-inclusion \Frechet Routes (EFR)}
We modify the map-matching algorithm to find a route that must include a specific edge $e=(u,v)$ with high residual flow. 
In fact, we are attempting to find $k$ routes which include the $k$ edges with the highest residual flow and which each have positive residual flow.
To find a route within \Frechet distance $\varepsilon$ of the representative trajectory $T$ which contains edge $e$, we need to find two path in the free-space manifold $\mathcal{F}$: a path from the start to a white interval at $u$ and a path from a white interval at $v$ to the end. Moreover, the concatenation of these two paths with $e$ needs to be monotone in $\mathcal{F}$. To do so, we consider each white interval at $u$, decide if it is reachable from the start, and if so, continue from all possible white intervals at $v$.

EFR stops its search in the free-space manifold as soon as the begin/endpoint of the reconstructed route lies within $\varepsilon$ distance of the start/end of $T$. This might ignore flow on edges in the $\varepsilon$-vicinity of the start and end of $T$. Hence, we greedily add suitable edges to the ends of the route, taking care not to introduce cycles and to not decrease the average amount of residual flow per edge in the route.

\subparagraph{Weighted \Frechet Routes (WFR)}
EFR uses the residual flow only to indicate the top $k$ routes. We further modify the map-matching algorithm to find routes that generally include edges with high residual flow, that is, high weight. To do so, we maintain a sorted list of weights with each white interval. Let $\tau$ be the parameter of trajectory $T$ (intuitively, the height values of the free space manifold $\mathcal{F}$). Each entry in the list is a tuple $(\tau,\psi)$ such that there is a path through $\mathcal{F}$ with weight at least $\psi$ (the value of $\psi$ depends on the execution order of Dijkstra's algorithm and is hence only a lower bound). We prune tuples $(\tau,\psi)$ whenever there is another tuple $(\tau',\psi')$ where $\tau'<\tau$ and $\psi' \geq \psi$. We maintain the weights as we execute the map matching algorithm. In principle, we can construct a high-weight route at the end of the algorithm. However, note that the road network $\RoadNet$ can contain cycles which result in cycles between white intervals. Hence we construct a high-weight route explicitly via back-tracking, using each tuple at most once. Note that the resulting route may still contain cycles; we break only cycles in the dependency of white intervals.

\subsection{Min-cost flow}\label{subsec:mincost}

We now describe two heuristics that are based on min-cost flow and relax the strong realism constraint of \Frechet Routes. On a high level, both heuristics follow the same approach: first, we solve the min-cost flow problem for the flow field (guided by the representatives to a certain degree) and then we heuristically compute a reconstruction from the resulting flow. Note that our cost function is the (quadratic) deviation of the reconstructed flow from the flow field, which differs from the standard (linear) cost function for min-cost flows.

\subparagraph{Multi-commodity min-cost flow (MCMCF)}
For each representative trajectory $T $ we construct a subgraph $G(T)$ of $\RoadNet$ consisting of all vertices and edges within distance $\varepsilon$ of $T$. Vertices that are within distance $\varepsilon$ from the start or end of $T$ can act as sources or sinks of a flow in $G(T)$. Each representative trajectory hence induces a single (min-cost) flow problem. By overlapping the graphs $G(T)$ for all $T \in \TrajSet$, we essentially construct a \emph{multi-commodity min-cost flow} problem on $\RoadNet$, where each trajectory $T$ has an associated commodity. We can solve the resulting MCMCF using standard software packages (see Section~\ref{sec:experiments}).

\subparagraph{Global min-cost flow (GMCF)}
We retain only the sources and sinks of MCMCF and otherwise impose no restriction on the flow. This results in a min-cost flow problem over the entire road network $\RoadNet$, which is essentially agnostic to the representative trajectories.

\subparagraph{Heuristic path reconstruction}
The result of either min-cost flow approach is an edge flow per commodity or over the complete road network. Our goal is to approximate these flows via a reconstruction that may use non-simple paths.
For each commodity (or the complete network) we first compute a ``path flows-cycle flows'' decomposition~\cite{ahuja2017network}, which is equivalent to the edge flows.
We can directly add the resulting path flows to our basis. The cycle flows, however, generally are not correct source-sink paths. 
We observe that cycle flows which are disjoint from all path flows cannot be close to any of the representative trajectories; we hence exclude them from the basis. All other cycle flows we greedily merge with one of the path-flows which overlap it at one or more vertices. If the path flow is higher than the cycle flow, then we reduce it to be equal to the cycle flow. If the path flow is lower than the cycle flow, then we traverse the cycle multiple times, using the path flow, to create a consistent (non-simple) route, rounding where necessary.

\section{Experiments}
\label{sec:experiments}

In this section we evaluate and compare the various heuristics of the previous section. Particularly, we investigate the  extensions of our \Frechet{}-Routes methods, the effect of the parameters such as $\varepsilon$, and compare \Frechet Routes to the min-cost-flow-based methods.

We implemented all algorithms in C++ using Boost and MoveTK\footnote{\href{https://movetk.win.tue.nl}{https://movetk.win.tue.nl}}. For flow problems and determining coefficients $c$, we use IBM ILOG CPLEX 12.9. We ran all experiments single-threaded on Ubuntu 18.04, on an Intel(R) Xeon(R) Gold 5118 CPU @ 2.30GHz.

\begin{figure}[b]
    \centering
    \includegraphics[scale=1.0]{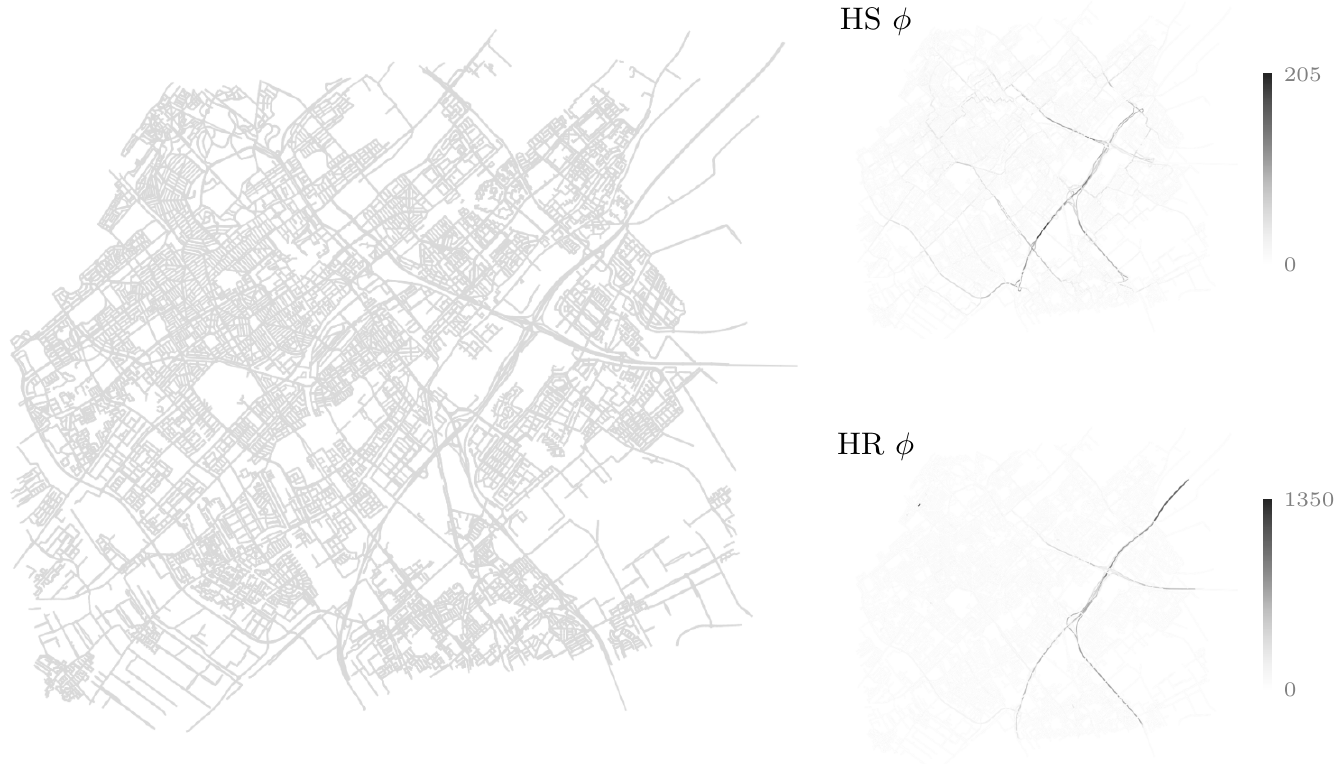}
    \caption{(Left) Road network of The Hague used in our experiments, $60\,277$ vertices and $100\,654$ edges. (Right) Input flow fields for HS and HR.}
    \label{fig:datanetwork}
\end{figure}

\subparagraph{Data}
To evaluate how well representative trajectories assist in route reconstruction beyond the given sample, we require a ground truth: all driven routes $\Solution^*$ that together define the flow field $\phi$. 
We use as network $\RoadNet$ the roads surrounding The Hague (the Netherlands) extracted from OpenStreetMap~\cite{OpenStreetMap}; see Fig.~\ref{fig:datanetwork}.
We create two datasets based on a complete trajectory set $\TrajSet^*$ from which we can derive ground truth $\Solution^*$:
\begin{description}
    \item[HS:] $\TrajSet^* = \Solution^*$ consists of $5\,000$ shortest paths between random locations in $\RoadNet$; variation is created by randomly perturbing every edge length with a value in the interval $[0,\gamma]$, for some parameter $\gamma \geq 0$, separately for each shortest-path computation; we use $\gamma=500m$. 
    \item[HR:] We use a set $\TrajSet^*$ of $11\,445$ real-world trajectories provided by HERE Technologies\footnote{\href{https://www.here.com}{https://www.here.com}} in the same area. We map-match $\TrajSet^*$ to $\RoadNet$ to obtain $\Solution^*$. To avoid bias, we use \cite{yang2018fast} instead of \cite{alt2003matching}, as the latter relies on the \Frechet distance and is the basis for our \Frechet Routes.
\end{description}
We derive the flow field $\phi$ for $\RoadNet$ by counting the number of occurrences of each edge in $\Solution^*$.
The representative trajectories $\TrajSet$ are sampled from $\TrajSet^*$, using $\alpha  > 0$ such that $|\TrajSet| = \lceil\alpha |\TrajSet^*|\rceil$.

\subparagraph{Measures}
Based on the rationale that we want to measure beyond the provided representatives, we quantify the performance of our heuristics via the following measures:
\begin{description}
    \item[flow deviation:] how well does the result represent the flow? This is measured as the sum of squared differences between input and reconstructed flow, over all edges (see Section~\ref{ssec:reconstructedroutes}).
    \item[realism:] how realistic are the routes? Following Section~\ref{ssec:realism}, we measure the average \Frechet distance from constructed route to its closest trajectory in $\TrajSet^*$, weighted by the coefficients.
    \item[coverage:] how well is all behavior captured by the result? This is measured as the average \Frechet distance from each route in ground truth $\Solution^*$ to the closest reconstructed route.
    \item[complexity:] the number of reconstructed routes.
    \item[running time:] the efficiency of the method,  measured as total computation time.
\end{description}

For all measures, lower values indicate better performance.
We note that finding a subset of $\TrajSet$ with optimal coverage is NP-hard, via a reduction from dominating set for unit-disk graphs \cite{masuyama1981computational}.
Contrasting HS, HR does not admit a solution with perfect coverage as the trajectories are not aligned to the road network. To indicate the distortion inherent in the data due to map matching, we visualize the distortion between $\TrajSet^*$ to $\Solution^*$ for HR in Fig.~\ref{fig:realdistortion}. 

\begin{figure}[h]
    \centering
    \includegraphics{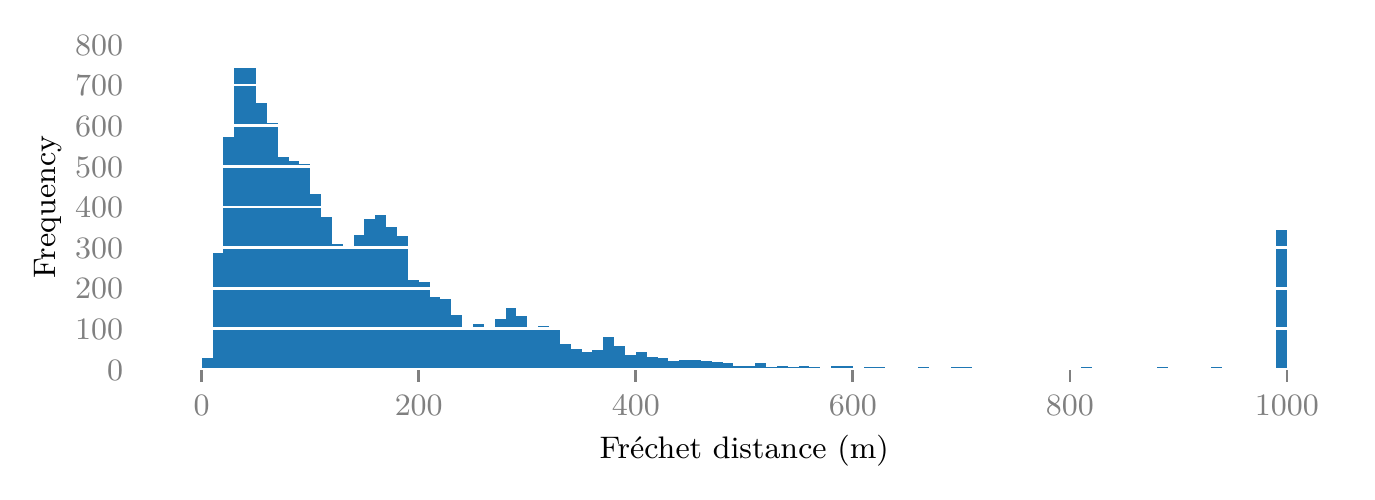}
    \caption{Histogram of \Frechet distances between all $T\in\TrajSet^*$ and their routes in $\Solution^*$ for HR. The rightmost bar aggregates higher values, which is approximately $3$\% of the trajectories.}
    \label{fig:realdistortion}
\end{figure}

\subsection{\Frechet Routes}
\label{sec:extensions}

Here we investigate our \Frechet{}-Routes algorithm, in terms of candidate-generation methods and its parameters. Throughout this section, we keep $\varepsilon$ fixed at $100m$ and run each trial with seven random samples and average the results.

\begin{figure}[p]
    \centering
    \includegraphics{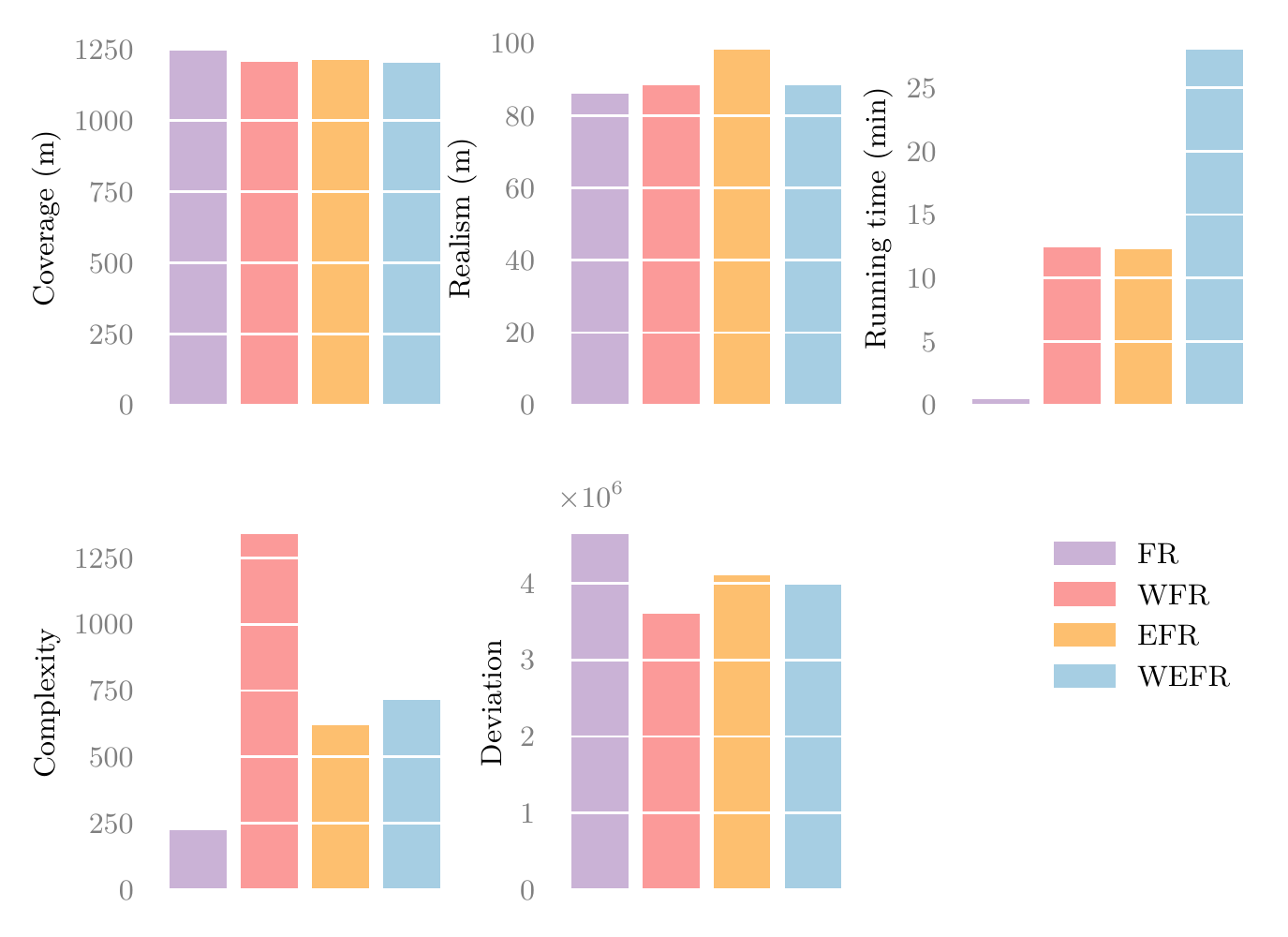}
    \hrule
    \centering
    \includegraphics{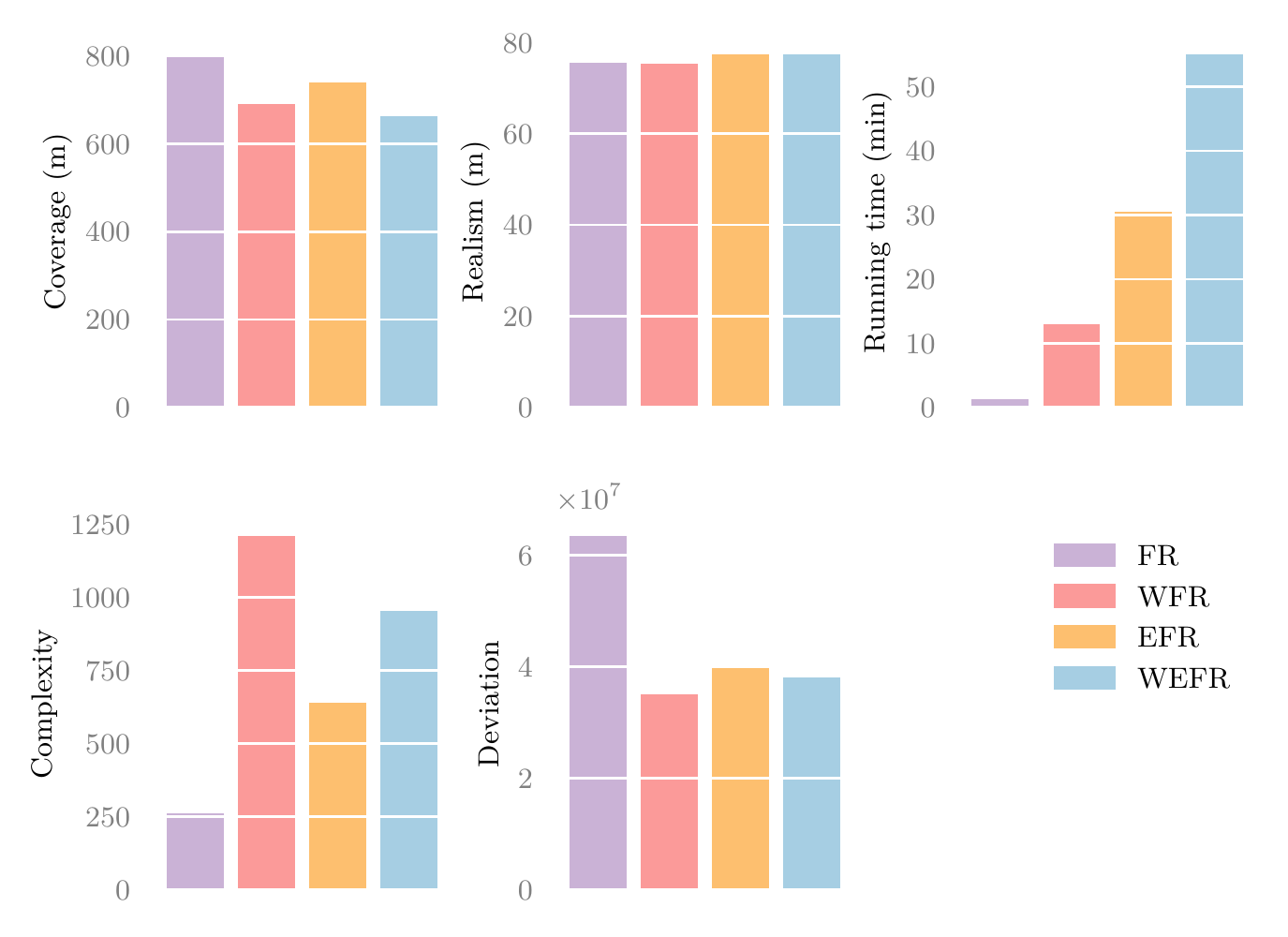}
   \caption{Results for HS (top) and HR (bottom) using variants of \Frechet Routes.}
    \label{fig:extensions-comparison}
\end{figure}

\subparagraph{Candidate generation}
We aim to investigate EFR and WEFR to generate basis routes.
Specifically, we run four variants: with weighted routes (WFR), with edge-inclusion (EFR), with both (WEFR) and without either extension (FR).
We use $\alpha=0.05$, $i_\text{max} = 8$ and $k =2$. 

Fig.~\ref{fig:extensions-comparison} summarizes the results for both datasets. We see fairly similar patterns between them. Compared to FR, the extensions have a mild positive effect on coverage and a strong positive effect deviation, but slightly deteriorate realism and increase complexity. Interestingly, EFR results in worse realism compared to WFR, but WFR causes a considerable increase in complexity. WEFR seems to combine the best of both, albeit at increased computation time. 

\begin{figure}[p]
    \centering
    \includegraphics{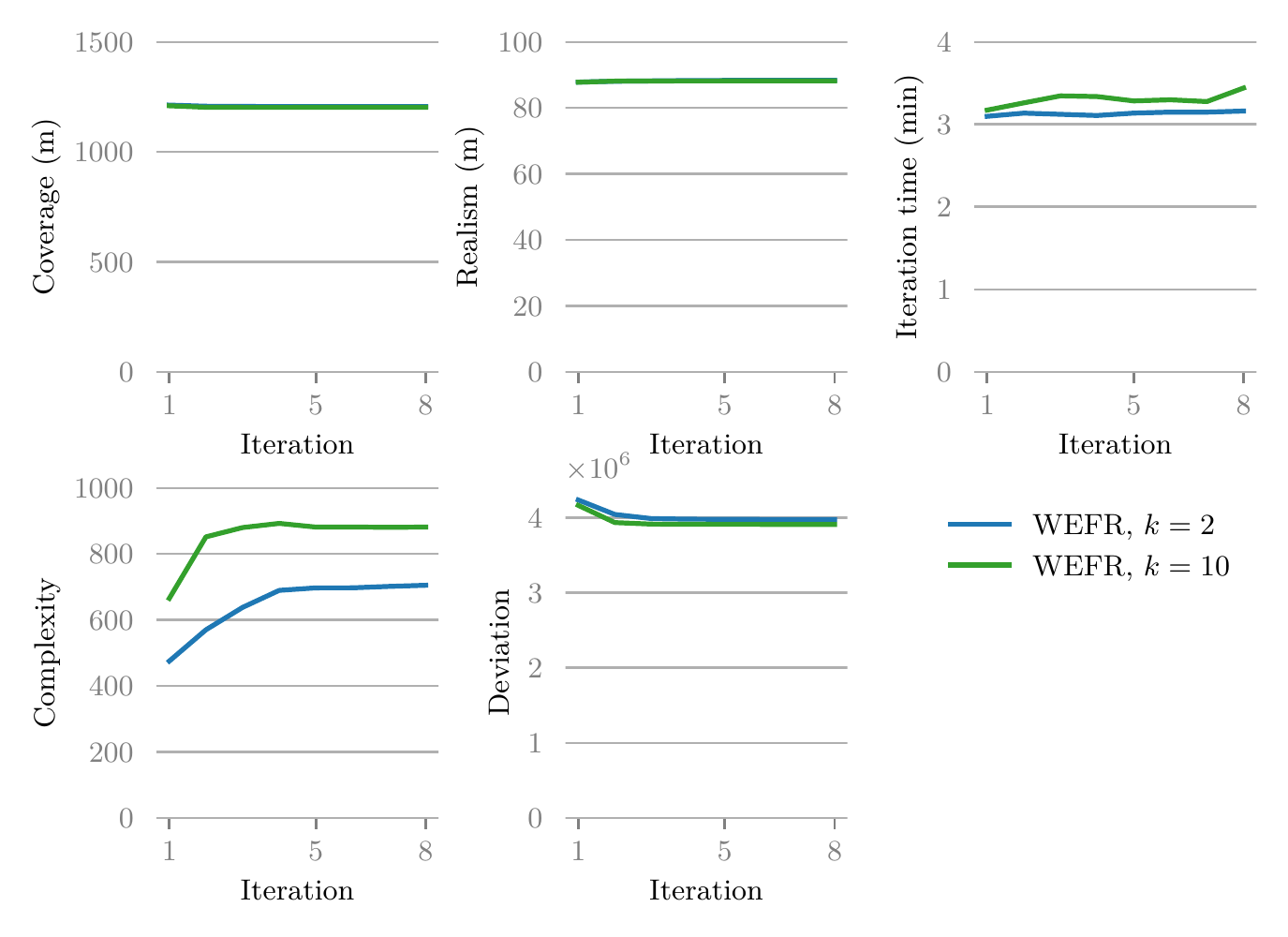}
    \hrule
    \centering
    \includegraphics{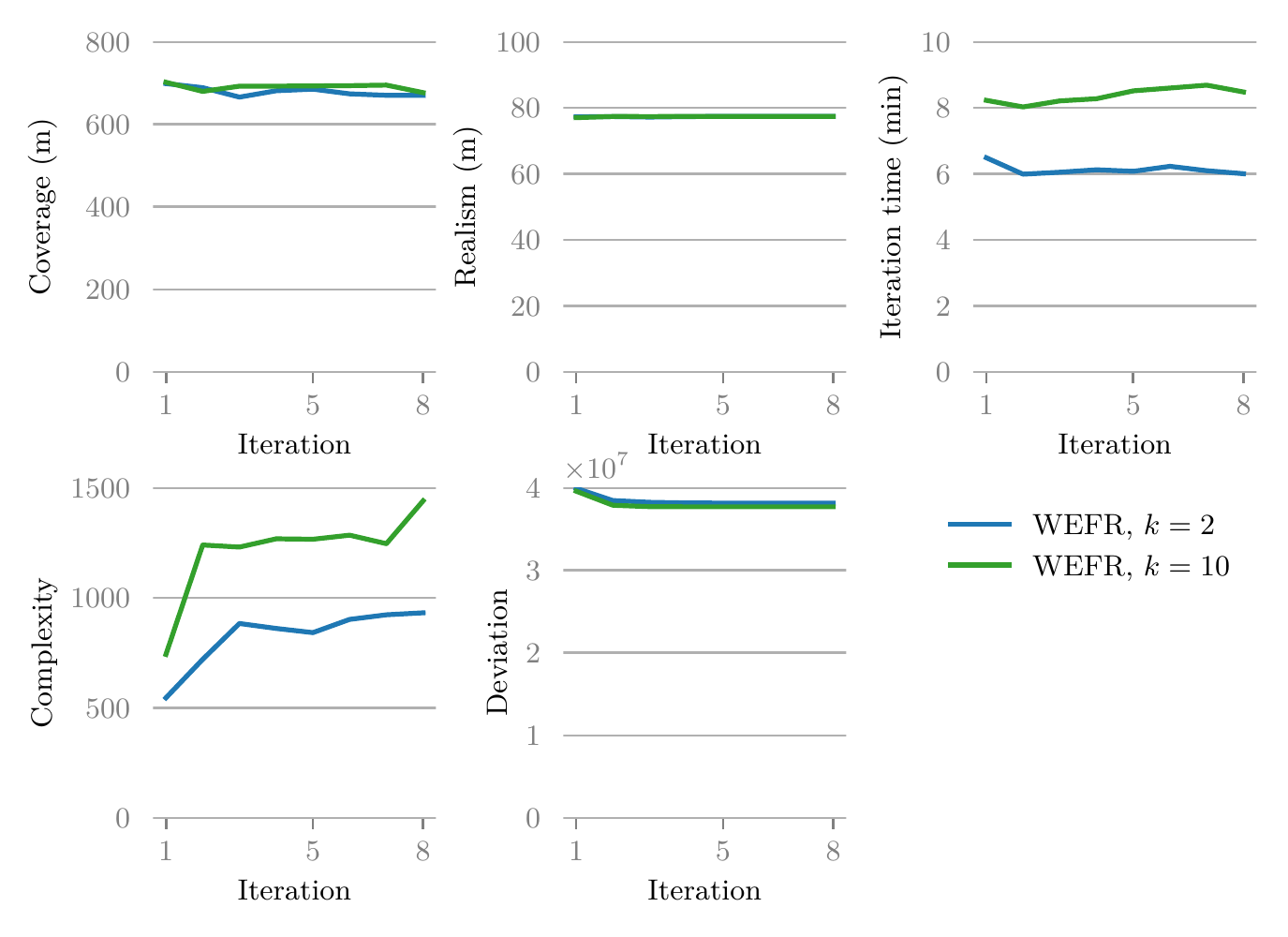}
    \caption{Results for WEFR per iteration, with $k=2$ and $k=10$ on HS (top) and HR (bottom).}
    \label{fig:extra-hittingpath-data}
\end{figure}

\subparagraph{Iterations and edge inclusion}
\label{sec:it-edge-incl}
Focusing now on WEFR, we investigate the effect of $i_\text{max}$, the number of iterations, and of $k$, the number of paths generated by edge inclusion. To do so, we run our algorithm with $k \in \{2,10\}$ on each dataset for $i_\text{max} = 8$ iterations and record the reconstruction after each iteration. As before, we use $\alpha=0.05$. 
Fig.~\ref{fig:extra-hittingpath-data} illustrates the results. For both datasets we see that the time spent per iteration remains roughly similar, even though the complexity tends to increase per iteration. Whereas realism remains mostly constant, iterating has a mildly positive effect on coverage and on deviation. After the first $5$ iterations this stabilizes, especially for deviation. Compared to $k = 2$, $k = 10$ yields higher complexity and running time, but improves deviation. The somewhat unstable coverage for HR is likely to be attributed to the use of map matching to obtain a ground truth.

\subparagraph{Sampling rate}
\label{sec:sampling-rate}
We now vary $\alpha$ using values in $\{ 0.05, 0.1, 0.15, 0.2, 0.25 \}$. This mimics different degrees of completeness of the representative trajectories relative to the overall traffic. 
Generally, we may expect that solution quality increases as we have more representative trajectories.
Based on the above, we use $k = 2$ but reduce $i_\text{max}$ to 5, since a larger sample implies that more candidates are generated per iteration.

\begin{figure}[t]
    \centering
    \includegraphics{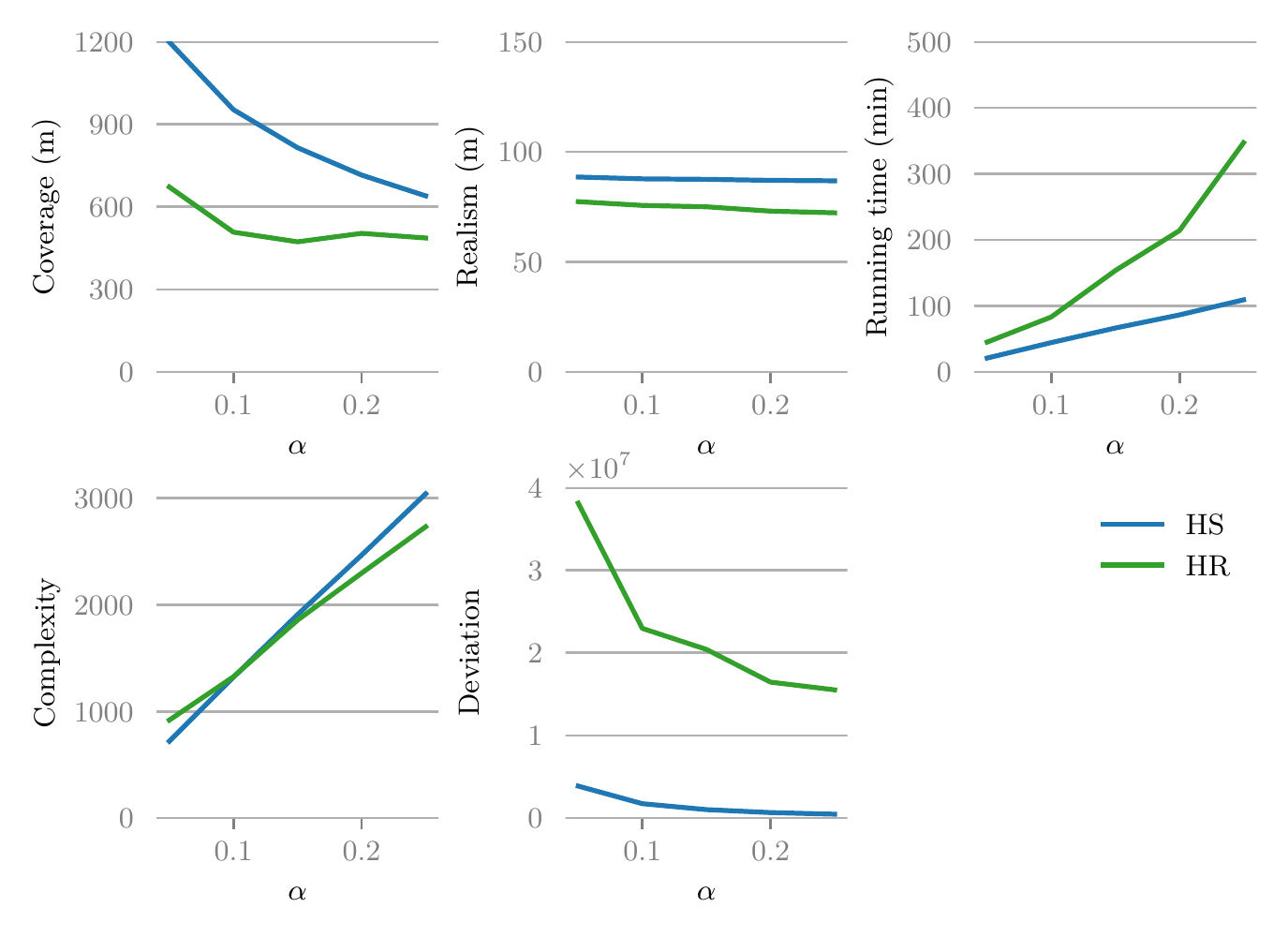}
    \caption{Results of WEFR for both datasets for varying sampling rate $\alpha$. }
    \label{fig:sampling-rate}
\end{figure}

The results are shown in Fig.~\ref{fig:sampling-rate}.
As expected, increasing sampling rate yields better deviation and coverage at the expense of longer computation times and higher complexity. Note that realism is not affected significantly, since it measures per reconstructed route. We see a stronger effect on coverage for HS and a stronger effect on deviation for HR. We attribute both to the nature of the data: whereas the synthetic data is spread over the network, the real data is focused on main roads (see Fig.~\ref{fig:datanetwork}). As such, additional trajectories are more likely to add a very different representative for synthetic data, but at the same time this requires a higher sample to be able to explain the resulting flow well.

\FloatBarrier
\subsection{Comparing \Frechet Routes to min-cost-flow methods}

\begin{figure}[b]
    \centering
    \includegraphics{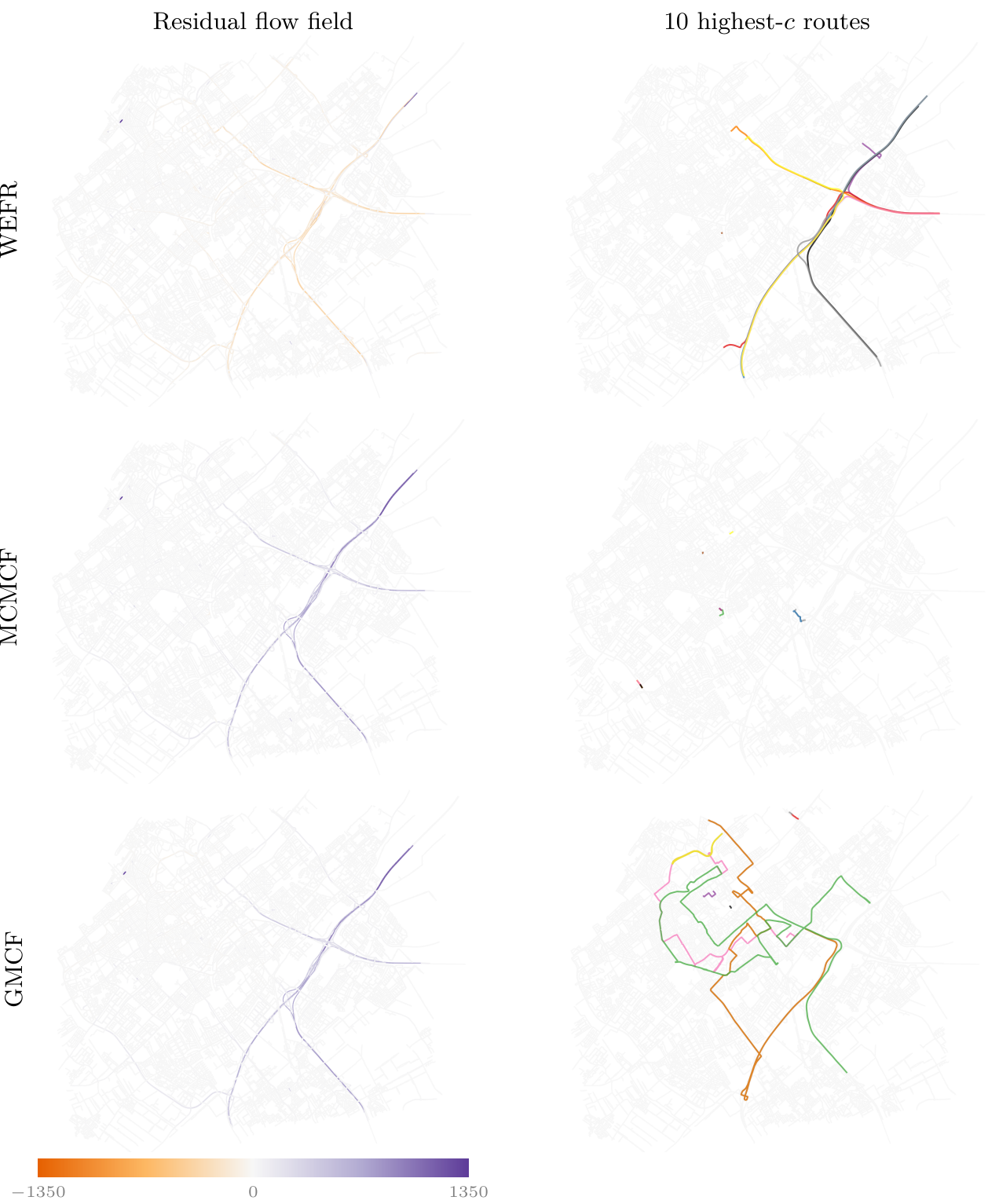}
    \caption{Residual flow field (left) and the 10 highest-coefficient reconstructed routes (right) for one of the samples of HR with $\alpha=0.05$ for each technique with $\varepsilon=100m$, $i_\text{max}=8$ and $k=2$.}
\label{fig:top10-examples}
\end{figure}

We now compare WEFR to GMCF and MCMCF. For these three methods, Fig.~\ref{fig:top10-examples} visualizes the flow deviation as well as the ten reconstructed routes with the highest coefficient $c$. Though the overall deviation of GMCF ($1.8\times10^7$) is lower than that of MCMCF ($3.7\times10^7$) and WEFR ($3.8\times10^7$), we observe that WEFR shows smaller places of flow deviation with a tendency to overrepresent the flow field, rather spreading out the deviation across the network, whereas GMCF and MCMCF have high deviation of underrepresentation concentrated along major routes with considerable traffic. 
Of further note is that the most contributing routes for WEFR look like actual routes, whereas the other methods give fairly unintuitive routes: they typically are either very short (MCMCF) or follow an unrealistic path (GMCF).

To further investigate, we look at how this comparison depends on threshold $\varepsilon$. We vary $\varepsilon$ using values $\Set{10\text{m},20\text{m},50\text{m},100\text{m},150\text{m},200\text{m},250\text{m}}$, while keeping $\alpha = 0.05$.
For WEFR we use $i_\text{max} = 8$ and $k = 2$, based on the previous section.
As before, we run each combination seven times and average the results.
Since the complexity of MCMCF and GMCF is orders of magnitude larger than that of WEFR, we analyze realism and coverage for the $2\,500$ routes with highest coefficient. We argue that this paints a better picture of similarity to the ground truth, as the many small weight paths do not give a concise description of the flow field. Moreover, this limit is still above the complexity of WEFR, thereby giving the competitors a slight advantage.

Fig.~\ref{fig:complexity-deviation-time} shows the results. We generally see the same patterns between the datasets: coverage and deviation improve as $\varepsilon$ grows. Realism deteriorates for WEFR and MCMCF, but improves for GMCF -- although it remains considerably larger than the other methods. We attribute this to the increased flexibility offered by having more starting vertices in $\RoadNet$ available. Deviation for GMCF is lower than WEFR and MCMCF.
MCMCF and WEFR behave similarly in coverage, realism and deviation. However, MCMCF has significantly higher complexity -- approaching the $O(E)$ upperbound -- whereas WEFR has higher running time for large $\varepsilon$. This suggests that the realism is somewhat inherent in the flow information, further supporting the case for complementarity of these data sources.
Finally, we note that the pattern in coverage for GMCF and for WEFR and MCMCF seem to be opposite for the two datasets. We attribute this to the need for map matching in HR which causes deviations between the flow field and the representative trajectories. This suggests that, although GMCF seems to perform well for large $\varepsilon$, this does not generalize to real data.

\begin{figure}[p]
    \centering
    \includegraphics{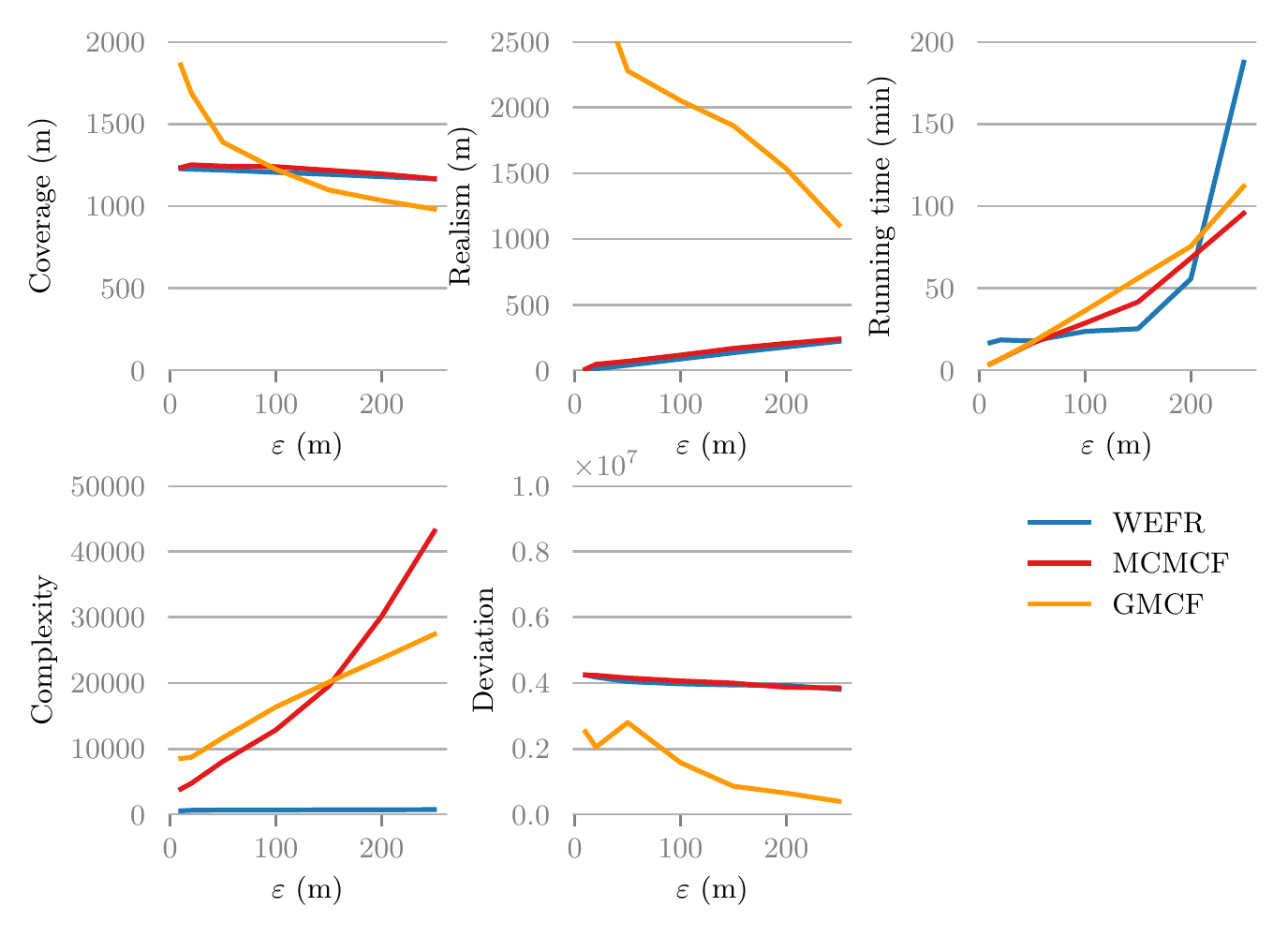}
    \hrule
    \centering
    \includegraphics{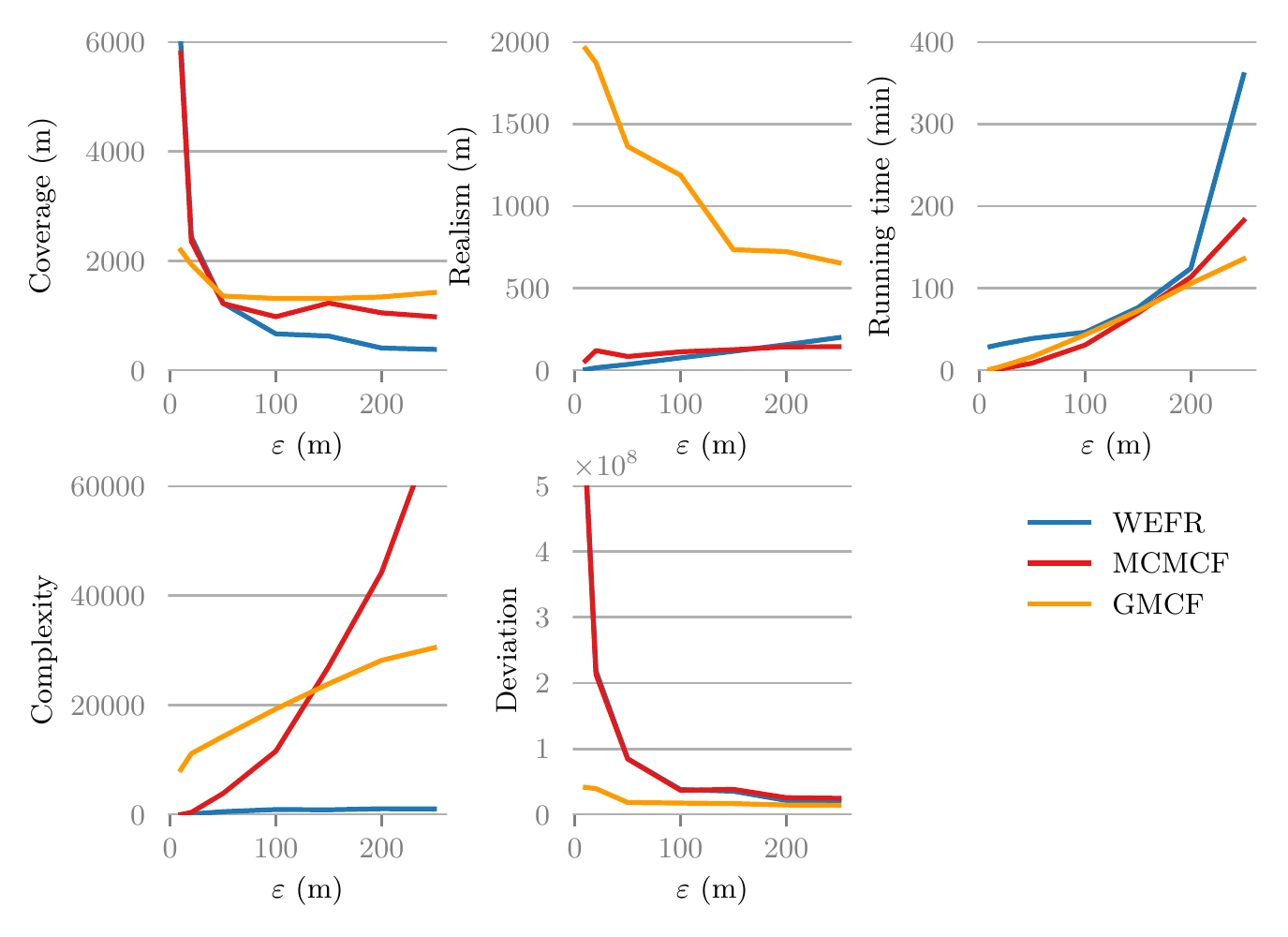}
   \caption{Results for our methods on HS (top) and HR (bottom) for varying $\varepsilon$.}
   \label{fig:complexity-deviation-time}
\end{figure}

\section{Discussion} \label{sec:discussion}
Our work shows that the studied problem is challenging even in ``simple'' forms, but our evaluation shows promising results. As such, it leads to various avenues for further research.

\subparagraph{Flow data}
We considered only flow data that is complete and static: every edge is observed and all data is aggregated into a single value. This is, however, typically not the case. Can we extend our techniques to cope well with unobserved edges and time-varying data? The challenge lies in finding good candidate routes despite the lack of flow information to guide this process, and coping with a temporal dimension that adds computational complexity. 

\subparagraph{Representatives}
Our experiments show that more representatives can improve the performance of route reconstruction. Yet, this also comes at a computational cost. However, this is mostly as it diversifies behavior: adding very similar representatives will not improve the results drastically, if at all. As our methods do not require that the representatives are actually part of the traffic generating the flow data, we could preprocess representatives using clustering and central trajectories to reduce computation time of our algorithms. Furthermore, we could use the information from such clustering methods, or from map-matching accuracy, to vary the threshold $\varepsilon$ per representative, to relate realism to the uncertainty in the data. We leave to future work to investigate how such techniques affect efficiency and quality.

\subparagraph{Reconstruction}
Though MCMCF does not guarantee a bound on the \Frechet distance, it performs similarly to WEFR in realism albeit with a very large basis. We could postprocess paths to allow only those that are realistic or attempt to incorporate this into the route-reconstruction phase, thereby reducing the complexity.
We modeled realism using the \Frechet distance. However, such paths may still traverse roads multiple times, reducing how intuitive such routes are. As strict simplicity makes the problem significantly harder, we may want to consider other models to further restrict what determines whether a route is realistic.

\bibliography{references}

\begin{thebibliography}{10}

\bibitem{ahuja2017network}
Ravindra~K Ahuja.
\newblock {\em Network Flows: Theory, Algorithms, and Applications}.
\newblock Pearson Education, 1993.

\bibitem{alt2003matching}
Helmut Alt, Alon Efrat, G{\"u}nter Rote, and Carola Wenk.
\newblock Matching planar maps.
\newblock {\em Journal of Algorithms}, 49(2):262--283, 2003.

\bibitem{alt1995computing}
Helmut Alt and Michael Godau.
\newblock Computing the fr{\'e}chet distance between two polygonal curves.
\newblock {\em International Journal of Computational Geometry \&
  Applications}, 5(01n02):75--91, 1995.

\bibitem{berndt1994using}
Donald~J Berndt and James Clifford.
\newblock Using dynamic time warping to find patterns in time series.
\newblock In {\em KDD workshop}, volume~10, pages 359--370. Seattle, WA, USA:,
  1994.

\bibitem{bouts2016mapping}
Quirijn~W Bouts, Irina~Irina Kostitsyna, Marc~van Kreveld, Wouter Meulemans,
  Willem Sonke, and Kevin Verbeek.
\newblock Mapping polygons to the grid with small {H}ausdorff and {F}r{\'e}chet
  distance.
\newblock In {\em Proceedings of the 24th Annual European Symposium on
  Algorithms (ESA 2016)}, volume~57 of {\em LIPIcs}, pages 22--1, 2016.

\bibitem{bro1997fast}
Rasmus Bro and Sijmen De~Jong.
\newblock A fast non-negativity-constrained least squares algorithm.
\newblock {\em Journal of Chemometrics: A Journal of the Chemometrics Society},
  11(5):393--401, 1997.

\bibitem{cao2013bilevel}
Peng Cao, Tomio Miwa, Toshiyuki Yamamoto, and Takayuki Morikawa.
\newblock Bilevel generalized least squares estimation of dynamic
  origin--destination matrix for urban network with probe vehicle data.
\newblock {\em Transportation research record}, 2333(1):66--73, 2013.

\bibitem{cascetta1984estimation}
Ennio Cascetta.
\newblock Estimation of trip matrices from traffic counts and survey data: a
  generalized least squares estimator.
\newblock {\em Transportation Research Part B: Methodological},
  18(4-5):289--299, 1984.

\bibitem{castillo2008observability}
Enrique Castillo, Antonio~J Conejo, Jos{\'e}~Mar{\'\i}a Men{\'e}ndez, and Pilar
  Jim{\'e}nez.
\newblock The observability problem in traffic network models.
\newblock {\em Computer-Aided Civil and Infrastructure Engineering},
  23(3):208--222, 2008.

\bibitem{duckham2016modeling}
Matt Duckham, Marc~J. van Kreveld, Ross Purves, Bettina Speckmann, Yaguang Tao,
  Kevin Verbeek, and Jo~Wood.
\newblock Modeling checkpoint-based movement with the earth mover's distance.
\newblock In {\em Proceedings of GIScience}, pages 225--239, 2016.

\bibitem{hartman2012split}
Tzvika Hartman, Avinatan Hassidim, Haim Kaplan, Danny Raz, and Michal Segalov.
\newblock How to split a flow?
\newblock In {\em 2012 Proceedings IEEE INFOCOM}, pages 828--836. IEEE, 2012.

\bibitem{HaunertB-2012}
J.-H. Haunert and B.~Budig.
\newblock An algorithm for map matching given incomplete road data.
\newblock In {\em Proc. 20th ACM SIGSPATIAL GIS}, pages 510--513, 2012.

\bibitem{jagadeesh2017online}
George~R Jagadeesh and Thambipillai Srikanthan.
\newblock Online map-matching of noisy and sparse location data with hidden
  markov and route choice models.
\newblock {\em IEEE Transactions on Intelligent Transportation Systems},
  18(9):2423--2434, 2017.

\bibitem{kim2010tackling}
Dongmin Kim, Suvrit Sra, and Inderjit~S Dhillon.
\newblock Tackling box-constrained optimization via a new projected
  quasi-newton approach.
\newblock {\em SIAM Journal on Scientific Computing}, 32(6):3548--3563, 2010.

\bibitem{kloster2018practical}
Kyle Kloster, Philipp Kuinke, Michael~P O'Brien, Felix Reidl,
  Fernando~S{\'a}nchez Villaamil, Blair~D Sullivan, and Andrew van~der Poel.
\newblock A practical {FPT} algorithm for flow decomposition and transcript
  assembly.
\newblock In {\em Proceedings of the 20th Workshop on Algorithm Engineering and
  Experiments (ALENEX 2018)}, pages 75--86. SIAM, 2018.

\bibitem{loffler2017discretized}
Maarten L{\"o}ffler and Wouter Meulemans.
\newblock Discretized approaches to schematization.
\newblock In {\em Proceedings of the 29th Canadian Conference on Computational
  Geometry}, 2017.

\bibitem{masuyama1981computational}
Shigeru Masuyama, Toshihide Ibaraki, and Toshiharu Hasegawa.
\newblock The computational complexity of the m-center problems on the plane.
\newblock {\em IEICE TRANSACTIONS (1976-1990)}, 64(2):57--64, 1981.

\bibitem{newson2009hidden}
Paul Newson and John Krumm.
\newblock Hidden markov map matching through noise and sparseness.
\newblock In {\em Proceedings of the 17th ACM SIGSPATIAL international
  conference on advances in geographic information systems}, pages 336--343,
  2009.

\bibitem{OpenStreetMap}
{OpenStreetMap contributors}.
\newblock {Planet dump retrieved from https://planet.osm.org }.
\newblock \url{ https://www.openstreetmap.org }, 2020.

\bibitem{quddus2007current}
Mohammed~A Quddus, Washington~Y Ochieng, and Robert~B Noland.
\newblock Current map-matching algorithms for transport applications:
  State-of-the art and future research directions.
\newblock {\em Transportation Research Part C: Emerging Technologies},
  15(5):312--328, 2007.

\bibitem{vatinlen2008simple}
Benedicte Vatinlen, Fabrice Chauvet, Philippe Chr{\'e}tienne, and Philippe
  Mahey.
\newblock Simple bounds and greedy algorithms for decomposing a flow into a
  minimal set of paths.
\newblock {\em European Journal of Operational Research}, 185(3):1390--1401,
  2008.

\bibitem{vegh2016strongly}
L{\'a}szl{\'o}~A V{\'e}gh.
\newblock A strongly polynomial algorithm for a class of minimum-cost flow
  problems with separable convex objectives.
\newblock {\em SIAM Journal on Computing}, 45(5):1729--1761, 2016.

\bibitem{WenkSP-2006}
C.~Wenk, R.~Salas, and D.~Pfoser.
\newblock Addressing the need for map-matching speed: Localizing global
  curve-matching algorithms.
\newblock In {\em Proc. 18th Int. Conf. on Sci. \& Stat. Database Management},
  pages 379--388, 2006.

\bibitem{yang2018fast}
Can Yang and Gyozo Gidofalvi.
\newblock Fast map matching, an algorithm integrating hidden markov model with
  precomputation.
\newblock {\em International Journal of Geographical Information Science},
  32(3):547--570, 2018.

\end{thebibliography}

\clearpage
\appendix

\section{Omitted proofs}
\label{app:hard}

\absdeviation*

\begin{proof}
We show that it is NP-hard to determine whether there exists an $(s,t)$-reconstruction with simple paths such that the absolute deviation $\Delta_\text{abs}(\Solution,c,\phi) \leq \delta$ for some $\delta \geq 0$. 

Our proof uses a reduction from the longest-path problem: given a graph $\RoadNet' = (V',E')$, a source $s'$ and sink $t'$ in $V'$, and a threshold $L > 0$, decide whether $\RoadNet'$ admits a simple path of at least $L$ edges from $s'$ to $t'$.
We turn this into an instance of our problem as follows. We augment the network to $\RoadNet$ by adding a new vertex $s$ connected by a new path of $L-1$ edges to $s'$. For the sink, we use the same vertex, $t = t'$. The flow field $\phi$ has value $1$ for all edges in $E'$, and value $0$ otherwise. Finally, we set the deviation threshold at $\delta = |E'|-1$.
Our claim is that this instance admits an $(s,t)$-reconstruction with absolute deviation at most $\delta$, if and only if $\RoadNet'$ admits a simple path of length at least $L$ from $s'$ to $t'$.

Assume that $\RoadNet'$ admits a simple path of length at least $L$ from $s'$ to $t'$. Let $P'$ denote this path, and $P$ the route in $\RoadNet$ consisting of the path from $s$ to $s'$ concatenated with $P'$. Consider the reconstruction consisting only of $P$ with coefficient $1$. The result is an $(s,t)$-reconstruction with only simple paths by construction. For every edge $e$ originally from $P'$ we have that $|\phi(e)-\sum_{P\in\Solution} M(P,e) c(P)| = 0$; for every other edge this value is $1$.
Thus, the absolute deviation $\Delta_\text{abs}(\Solution,c,\phi)$ is the number of edges along the path from $s$ to $s'$ plus the number number of edges in $\RoadNet'$ not covered by $P'$. Hence, this is at most $(L-1) + (|E'| - L) = |E'| - 1 = \delta$.

Now, assume that an $(s,t)$-reconstruction $(\Solution,c)$ exists with absolute deviation at most $\delta$. Specifically, we assume that $\Delta_\text{abs}(\Solution \setminus \{ P \},c,\phi) > \Delta_\text{abs}(\Solution,c,\phi)$ for all routes $P \in \Solution$. In other words, removing any route increases the deviation.
Observe that a reconstruction without any routes would achieve deviation $|E'| > \delta$, thus the reconstruction must contain at least one route. Consider a route $P \in \Solution$. 
As $P$ must start at $s$ and end at $t$, and the flow field at all edges between $s$ and $s'$ is zero, we know that removing $P$ from the solution locally decreases the deviation by $(L-1) \cdot c(P)$. Since removing a path must increase deviation and the deviation at one edge can increase by at most $c(P)$, $P$ must contain at least $L$ edges originating from $\RoadNet'$.
Hence, the subpath of $P$ starting at $s'$ and ending at $t'$ has at least $L$ edges.
As a result, we conclude that $\RoadNet'$ has a simple path of length at least $L$.

Finally note that the construction can easily be performed in polynomial time, so the stated problem is NP-hard. 
\end{proof}

\sqrhard*

\begin{proof}
We again use a reduction from the longest path problem. Let the instance of the longest path problem consist of a graph $\RoadNet' = (V', E')$ along with source $s'$, sink $t'$, and threshold $L$. We augment $\RoadNet'$ to obtain $\RoadNet = (V, E)$ as follows. We first add a new path of $L-1$ edges from the new source vertex $s$ to $s'$, and we refer to this set of new edges as $E_1$. Furthermore, we add a new path of $L-1$ edges from $s'$ to $t'$, and we refer to this set of edges as $E_2$. Hence we have $E = E' \cup E_1 \cup E_2$. For the sink, we use the same vertex $t = t'$. For the flow field $\phi$ we set $\phi(e) = 0$ if $e \in E_1$, $\phi(e) = 1$ if $e \in E'$, and $\phi(e) = 2$ if $e \in E_2$. We now claim that the minimum deviation for the instance formed by $\RoadNet$ and $\phi$ is $|E|$ if the longest simple path in $\RoadNet'$ has length at most $L-1$, and the minimum deviation is strictly smaller than $|E|$ if there exists a path of length at least $L$ in $\RoadNet'$. This directly implies that the stated reconstruction problem is NP-hard.

First assume that the longest simple path in $\RoadNet'$ has length at most $L-1$. Consider the $(s,t)$-reconstruction $(\Solution, c)$ consisting of a single path $P^* \in \Solution$ with $P^* = E_1 \cup E_2$ and $c(P^*) = 1$. Note that the deviation of this reconstruction is exactly $|E|$. Now consider any simple path $P$ between $s$ and $t$ and add it to $\Solution$ with $c(P) = 0$ (this does not really change the decomposition). By definition, the derivative of the deviation with respect to the coefficient $c(P)$ is given by:
\begin{align*}
\frac{\partial \Delta(\Solution,c,f)}{\partial c(P)} &=  \sum_{e\in E} -2 M(P, e) (f(e) - \sum_{P_i \in \Solution} M(P_i, e) c(P_i)) \\
&= 2 \sum_{e\in E} M(P, e) (M(P^*, e) - f(e)) \\
&= 2 (|P \cap E_1| - |P \cap E_2| - |P \cap E'|)
\end{align*}
The last step in the above is due to $M(P^*, e) - f(e)$ being $1$ for $e \in E_1$ and $-1$ otherwise.
Note that $|P \cap E_1| = L-1$ by construction. As any simple path $P$ either passes through $\RoadNet'$ or is equal to $P^*$, precisely one of the other terms is zero. If $P = P^*$, then $|P \cap E_2| = L-1$ and thus the derivative is zero.
Otherwise, $P$ passes through $\RoadNet'$ and thus $|P \cap E'| \leq L-1$. Hence, the derivative with respect to $c(P)$ is always non-negative. Since $c(P)$ cannot become smaller than zero, we conclude that $(\Solution, c)$ is a local minimum for the deviation function. Since both the constraints and the deviation function are convex, this local minimum must also be the global minimum. 

Now assume that there exists a simple path $P'$ in $\RoadNet'$ with length at least $L$. As above, let $P^* = E_1 \cup E_2$, and let $P = E_1 \cup P'$. Consider the $(s,t)$-reconstruction $(\Solution, c)$ with $\Solution = \{P^*, P\}$, $c(P^*) = 1$, and $c(P) = \epsilon$. Note that the edges in $E_2$ contribute exactly $|E_2|$ to the deviation. For an edge $e \in E_1$, the deviation is $(1 + \epsilon)^2$, and for an edge $e \in P'$ the deviation is $(1 - \epsilon)^2$. In total, the edges in $E'$ contribute at most $|E'| - L + L (1 - \epsilon)^2$ to the deviation. We obtain the following total deviation over all edges:
\begin{align*}
\Delta(\Solution,c,f) &\leq |E_2| + |E'| - L + L (1 - \epsilon)^2 + (L - 1) (1 + \epsilon)^2 \\
&= |E| - 2L + 1 + L (1 - \epsilon)^2 + (L - 1) (1 + \epsilon)^2 \\
&= |E| - 2L + 1 + 2 (L - 1) (1 + \epsilon^2) + (1 - \epsilon)^2 \\
&= |E| - 1 + 2 (L - 1) \epsilon^2 + (1 - \epsilon)^2 \\
&= |E| + (2L - 2) \epsilon^2 - 2\epsilon + \epsilon^2 \\
&= |E| + \epsilon ((2L - 1) \epsilon - 2)
\end{align*}
Finally, by choosing $\epsilon = 1/(2L-1)$, we obtain that $\Delta(\Solution,c,f) = |E| - 1/(2L-1) < |E|$, which concludes the proof. 
\end{proof}

\section{\Frechet map-matching}
\label{app:map-matching}
Alt \etal.~\cite{alt2003matching} present an algorithm to map-match a trajectory $T=\Seq{p_1,\ldots,p_{l}}$ to a road network $\RoadNet$: the result is a route in $\RoadNet$ such that its \Frechet distance to $T$ is at most $\varepsilon$. This decision version can subsequently be used to find the route that minimizes the \Frechet distance to $T$, but we need only this decision algorithm.

Let the allowed \Frechet distance be fixed to $\varepsilon$, and let $T$ be parameterized on $[1,l]$ such that for $\tau \in [i,i+1]$, $T(\tau) = p_i + (\tau-i) (p_{i+1}-p_i)$. Furthermore, we assume for simplicity here that $\RoadNet$ is the subset of the overall network that is fully covered by the Minkowski sum of a disk of radius $\varepsilon$ with the trajectory $T$.

Trajectory $T$, network $\RoadNet$ and $\varepsilon$ define a free-space manifold $\mathcal{F}$. This manifold defined as the locations in the direct product space $(\tau,r) \in T \times \RoadNet$ such that $d(T(\tau),r) \leq \varepsilon$.
A monotone path through the free space on this free-space manifold from some point with $\tau = 1$ to some point with $\tau = l$ then matches to a route in $\RoadNet$ with \Frechet distance at most $\varepsilon$. Note that the monotonicity avoids moving backwards along an edge, but an edge may be visited multiple times.

First, for each vertex $v$ of $\RoadNet$, the algorithm computes a 1-dimensional free-space diagram $FD(v) \colon [1,l] \rightarrow \Set{0,1}$. In this diagram, intervals that map to value $1$ are called \emph{white intervals}. The white intervals mark the parameter values of $\tau$ on $T$ for which $v$ is within  Euclidean distance $\varepsilon$ and thus a potential match.

Then, for each edge $(u,v)$ in $\RoadNet$, the algorithm computes the left-right pointers for all white intervals $I$ of $FD(u)$. These left-right pointers mark a range of $FD(v)$ that can be reached from $I$ by a monotone path in the free space of the connecting 2-dimensional free-space strip defined by $(u,v)$ and $T$. Due to convexity any point in the free space of $FD(v)$ between these extremal left-right pointers can be reached, but note that this range may encompass multiple white intervals at $v$.

To determine whether a route exists within \Frechet distance $\varepsilon$, the algorithm now applies a sweep-line algorithm over the parameter space defined by $\tau$, starting at 1. For every vertex $v$, it maintains a list $C(v)$ of ranges in $FD(v)$ that can be reached by a monotone path in the free-space manifold defined by $\RoadNet$ and $T$. That is, the intersection of $C(v)$ with the white intervals of $FD(v)$ define reachable intervals. In particular, this list contains the ranges such that the last part of the matching path ends at or goes through the sweep-line value in the dimension of $\tau$. 

The events that the sweep line handles are when the start of a reachable white interval $I$ in some $C(u)$ is reached that was previously not discovered yet. This event is handled by updating all $C(v)$ for all $v$ such that $(u,v)$ in $E$, taking into account the new intervals that can be reached from $I$. These can be efficiently retrieved via the precomputed left-right pointers.

The sweep line stops as soon as a white interval is added to some $C(v)$ that contains parameter value $l$ or no white intervals are available anymore. The former implies the existence of a route within \Frechet distance $\varepsilon$, whereas the latter implies that such a route does not exist. 
To reconstruct the route, the algorithm keeps track of predecessor vertices whenever the white intervals in $C(v)$ are updated when handling a white interval $I$.
In total, the algorithm runs in $O(l (V+E))$ time, assuming efficient bookkeeping. 
\end{document}